\documentclass[11pt]{article}
\usepackage{anysize}
\usepackage{multicol}
\usepackage{tabularx}
\usepackage{comment}
\usepackage{graphicx}
\usepackage{subfigure}
\usepackage{amsmath,amsthm,amssymb,amsfonts,thmtools}
\usepackage{url}
\usepackage{enumerate}
\usepackage{latexsym,amssymb}
\usepackage{mathtools}
\usepackage{pgfplots}
\allowdisplaybreaks[4]
\marginsize{2.5cm}{2.5cm}{2.5cm}{2.5cm}

\def\E{\mathbb{E}}
\def\M{\mathcal{M}}
\newcommand{\costrate}{{\rm cr}}
\newcommand{\cp}{{\rm CP}}
\newcommand{\optim}{{\rm OPT}}
\newcommand{\algor}{{\rm ALG}}
\newcommand{\marker}{{\rm Marker}}
\newcommand{\roe}{{\rm RoE}}

\newcommand{\lru}{{\rm LRU}}

\newcommand{\plfu}{{\rm PLFU}}
\newcommand{\prob}{{\rm Pr}}

\newtheorem{claim}{Claim}
\newtheorem{theorem}[claim]{Theorem}
\newtheorem{lemma}[claim]{Lemma}
\newtheorem{proposition}[claim]{Proposition}
\newtheorem{corollary}[claim]{Corollary}

\date{}

\begin{document}

\title{Modeling Online Paging in Multi-Core Systems}
\author{Mathieu Mari\thanks{University of Montpellier, Montpellier, France. mari.mathieu.06@gmail.com}, Anish Mukherjee\thanks{University of Warwick, Coventry, England.  anish343@gmail.com}, Runtian Ren\thanks{IDEAS NCBR, Warsaw, Poland. renruntian@gmail.com}, Piotr Sankowski\thanks{University of Warsaw, IDEAS NCBR and MIM Solutions, Warsaw, Poland. piotr.sankowski@gmail.com}}
\maketitle

\begin{abstract}
Web requests are growing exponentially since the 90s due to the rapid development of the Internet.
This process was further accelerated by the introduction of cloud services. 
It has been observed statistically that memory or web requests generally follow power-law distribution (Breslau et al., INFOCOM'99). 
That is, the $i^{\text{th}}$ most popular web page is requested with a probability proportional to $1 / i^{\alpha}$ ($\alpha > 0$ is a constant). 
Furthermore, this study, which was performed more than 20 years ago, indicated Zipf-like behavior, i.e., that $\alpha \le 1$. 
Such laws are useful for designing caching algorithms in general and configuring proxy caches, and therefore it is important to understand how universal they are.

We revisit these stochastic models in the case of cache requests in the context of modern cloud systems. 
Surprisingly, the memory access traces coming from petabyte-size systems not only show that $\alpha$ can be bigger than one but also illustrate a shifted power-law distribution --- called Pareto type II or Lomax. 
These previously not reported phenomenon calls for statistical explanation. 
Our first contribution is a new statistical {\it multi-core power-law} model indicating that double-power law can be attributed to the presence of multiple cores running many virtual machines in parallel on such systems. 
We verify experimentally the applicability of this model using the Kolmogorov-Smirnov test (K-S test). 

The second contribution of this paper is a theoretical analysis indicating why LRU and LFU-based algorithms perform well in practice on data satisfying power-law or multi-core assumptions. 
We provide an explanation by studying the online paging problem in the stochastic input model, i.e., the input is a random sequence with each request independently drawn from a page set according to a distribution $\pi$.
To study the performance of LRU and Perfect LFU (PLFU, an ideal version of LFU), we use ratio-of-expectations ($\roe$), by comparing the expected cost of LRU or PLFU with the expected cost of the optimal offline solution OPT, both dealing with a random sequence of requests with sequence length large enough.
We derive formulas (as a function of the page probabilities in $\pi$) to upper bound their ratio-of-expectations, which help in establishing $\roe(\lru) = O(1)$ and $\roe(\plfu) = O(1)$ given the random sequence following power-law and multi-core power-law distributions. 
\end{abstract}

\section{Introduction}
\label{section-intro}
Nowadays, the public search online for news and information, anytime and anywhere.
Thanks to the rapid development of the Internet, billions of users have registered on Google, Facebook, Twitter, YouTube, etc, and thousands of requests from these users are sent to the cloud servers every second.
Optimizing the average response time for user requests, as a consequence, becomes a core issue for running the business of the platform. 
For this purpose, caching plays a central role --- a ``good'' caching replacement policy is essential for improving the efficiency of cloud servers.

%For example, {\em Least-Recently-Used} (LRU), one of the most famous caching replacement policies (even known as the ``golden rule''), is widely adopted to maintain such caches. 

It has been observed statistically that, in many contexts, the requests seen by the web proxy caches typically follow Zipf-like distribution \cite{breslau1999web, arlitt2000performance, chesire2001measurement, almeida2001analysis, sripanidkulchai2004analysis, cherkasova2004analysis, atikoglu2012workload, huang2013analysis, yang2020large}. 
The first of these papers defines Zipf-like distribution to be a power-law distribution with exponent smaller than $1$.
More formally speaking, the probability of a request being the $i^{\text{th}}$ most popular web page is proportional to $1 / i^{\alpha}$ (here $\alpha > 0$ is a constant).
One typical characteristic of this power-law distribution is the ``20/80'' rule (or ``10/90'' rule, depending on the value of $\alpha$), which means around 80\% of the requests refer to the 20\% most popular web pages.

In the first part of this paper, we revisit these stochastic models in the case of cache requests in the context of modern cloud systems. 
Surprisingly, the memory access traces coming from these petabyte-size systems not only show that $\alpha$ can be bigger than one but also illustrate a shifted power-law distribution --- called Pareto type II or Lomax. 
This is visualized in Figure \ref{fig:lomax}. 
For comparison please see perfect Pareto type I distributions reported in~\cite{breslau1999web}.

This previously not reported phenomenon calls for a statistical explanation. 
Such an explanation should be consistent with numerous results reported in~\cite{breslau1999web, arlitt2000performance, chesire2001measurement, almeida2001analysis, sripanidkulchai2004analysis, cherkasova2004analysis, atikoglu2012workload, huang2013analysis, yang2020large}. 
Up to our understanding, the cache trace data reported in these papers is not coming from servers supporting virtualization. 
In particular, most of these cache traces are pre-2000 or are coming from early 2000, when server virtualization was not yet widely adopted~\cite{RedHat,10.1145/3365199}. 
On the contrary, the data studied here comes from petabyte size systems supporting massive virtualization, e.g., running in parallel hundreds of virtual machines. 
Our first contribution is a new statistical {\it multi-core power-law} model indicating that the presence of Pareto type II can be attributed to the presence of multiple cores running many virtual machines in parallel on such systems. 
The more virtual machines are present in the system the more the distribution is flattened. 
We verify experimentally the applicability of this model using the Kolmogorov-Smirnov test (K-S test) and observe that it fits much better to the data than simple power-law distribution. 
As argued in~\cite{why-kstest-for-powerlaw} the proper method of analyzing power-law distributed data should involve a goodness-of-fit test, as otherwise can lead to wrong conclusions.

Due to the statistical observations that reveal skewed distributions, an intuitive idea is to use {\em Least-Frequently-Used} (LFU) cache replacement policy (a.k.a., online paging algorithm). Here one keeps the most popular web pages inside the cache, to maximize the request-hit ratio (and hence minimize the average response time).
In fact, many realistic web caching replacement policies are hybrid in adopting both LFU and LRU ({\em Least-Recently-Used}) ideas and they have been verified to perform rather well in practice (i.e., achieving high hit ratio) \cite{karedla1994caching, karakostas2002exploitation, podlipnig2003survey, einziger2017tinylfu}. We note 
that some existing papers, e.g. \cite{podlipnig2003survey}, include experiments on synthetic data sets generated according to Zipf-like distributions~\cite{breslau1999web}. In our opinion, future synthetic data sets should be extended with Pareto type II distributions with exponents higher than 1.

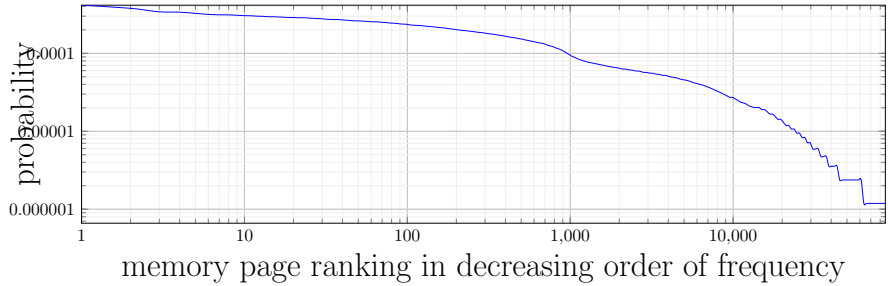
\begin{figure}
    \centering
    \begin{tikzpicture}[scale=0.6] 
\begin{axis}[
    xmode=log,
    log ticks with fixed point,
    ymode=log,
    log ticks with fixed point,
    xmin = 1, xmax = 86011,
    ymin = 0, ymax = 0.000415014,
    grid = both,
    minor tick num = 1,
    major grid style = {lightgray},
    minor grid style = {lightgray!25},
    width = 1.2\textwidth,
    height = 0.4\textwidth,
    xlabel = {memory page ranking in decreasing order of frequency},
    ylabel = {probability},
    label style={font=\huge}]

% data
\addplot[
    smooth,
    thin,
    blue,
] file {trace_data/probabilities.dat};

\end{axis}

\end{tikzpicture}
\caption{Frequency of memory page accesses versus page ranking. A power-law distribution should give a straight line on a doubly logarithmic scale.}
\label{fig:lomax}
\end{figure}

We stress that paging is one of the famous failures of worst-case analysis where it gives misleading or useless advice about which algorithm shall be used in practice -- as stated in the introduction section of the book by Roughgarden~\cite{roughgarden_2021}.
Thus, from a theoretical perspective, a natural question thus appears as follows: {\em is there any theoretical result to support the good performance on LFU or LRU when the sequence of requests follows power-law distribution?}  
In this paper, we for the first time provide an affirmative answer to this question.
\begin{itemize}
    \item[-] We model the problem as online paging in an ideal stochastic model, {\em where requests are generated independently by a given distribution $\pi$.} 
    \item[-] For LFU, we focus on one of its ideal version called {\em Perfect LFU} (abbreviated as PLFU, which counts the frequency over all the web pages and maintain the size-$k$ cache with $k$ most frequent pages). 
    \item[-] Under the stochastic model, we evaluate the performance of LRU and PLFU using {\em ratio-of-expectation} ($\roe$) by comparing the expected cost of PLFU with the expected cost of the offline optimal solution $\optim$ (both dealing with a length-$n$ sequence of requests). 
    \item[-] We derive three different formulas (each as a function of the page probabilities in the given distribution $\pi$) to upper bound $\roe(\plfu)$ and one formula to upper bound $\roe(\lru)$.
    \item[-] Finally, we show that when the distribution $\pi$ is a power-law or multi-core power-law distribution and the number of all pages is at least twice the cache size, our formulas only depend on $\alpha$. In other words, PLFU and LRU both achieve a constant ratio-of-expectation.\footnote{Note that $\alpha > 0$ is a constant to characterize the statistics property of the given sequence of requests, typically irrelevant to the cache size or the population of all pages.}
\end{itemize}
We remark that although it is unlikely that the requests are independent of each other and each follows a fixed distribution in practice, such a ``simple'' stochastic model is still powerful enough to provide theoretical indications explaining why the LFU and LRU-based online algorithms practically perform rather well. 

\paragraph{The paging problem.}
We now formally define the paging problem as follows. Consider a two-level memory system: a small fast memory (the cache) that can store $k$ pages and a large but slow memory (e.g., HDD) that can potentially hold infinitely many pages. A sequence of page requests needs to be processed in such a memory system: each request specifies a page and such a page has to be inside the cache for its processing. When a request arrives, if the requested page is already in the cache, then no cost is incurred; otherwise, a {\em page fault} occurs, i.e., one page in the cache must be evicted to make room for the requested page and some cost is incurred. A paging algorithm decides which page to evict when facing a page fault and the paging problem seeks an optimal algorithm to minimize the total cost incurred while processing a sequence of requests. 

\paragraph{Stochastic input model.}
We study the online paging problem in the stochastic input model, where each request (in the random sequence) is independently drawn from a set of pages $\M$ according to a distribution $\pi: \M \to (0, 1)$.
To evaluate the performance of an online paging algorithm $\algor$ we use {\em ratio-of-expectation} ($\roe$), by comparing the expected cost of $\algor$ with the expected cost of the offline optimal solution $\optim$.
To elaborate, we need the following definitions and notations. 

Given any sequence of requests $\sigma$, for any paging algorithm $\algor$, let $\algor(\sigma)$ denote the cost of $\algor$.
In the case that $\sigma$ is a length-$n$ random sequence with each request generated according to $\pi$, let $\E_{\sigma}^n[\algor(\sigma)]$ denote the expected cost of $\algor$ to deal with this sequence $\sigma$.\footnote{If $\algor$ is a randomized algorithm (which flips a coin to maintain the cache configuration), let $r$ denote the random coins used, and the expected cost of $\algor$ becomes $\E_{\sigma}^n\big[\mathbb{E}_r[\algor(\sigma, r)]\big]$.}
W.l.o.g., we assume that the $i^{\text{th}}$ request in the given sequence arrives at time $i$. 
In the online situation, an online paging algorithm needs to process request $t$ at time $t$, before the arrival of the next request.
The ratio-of-expectation of a paging algorithm $\algor$ is defined as
\begin{equation*}
\roe(\algor) = \overline{\displaystyle\lim_{n \to \infty}} \frac{\E_{\sigma}^n[\algor(\sigma)]}{\E_{\sigma}^n[\optim(\sigma)]}.
\end{equation*}
Here, $\optim(\sigma)$ denotes the cost of the offline optimal solution $\optim$ for $\sigma$.
We remark that the ratio-of-expectation criteria have been widely used to evaluate the performance of an online algorithm for Steiner forest \cite{garg2008stochastic}, perfect matching \cite{gupta2019stochastic, mathieu2023online}, $k$-Server \cite{dehghani2017stochastic}, set cover \cite{grandoni2008set} in the stochastic input model. We thus believe that it is also suited for evaluating the performance of online paging algorithms. 

\paragraph{Least-Frequently-Used.}
Now we elaborate on the LFU method by introducing two versions of LFU-based paging algorithms: {\em Perfect LFU} and {\em In-Cache LFU}. 
\begin{itemize}
    \item[-] Perfect LFU (PLFU) counts the page frequency over all the pages and maintains the size-$k$ cache with $k$ most frequent pages. At a page fault, the $k^{\text{th}}$ most popular page is evicted to make room for the incoming requested page; right after processing the request, the $k^{\text{th}}$ most popular page is added back inside the cache immediately -- a cost of 2 is thus incurred for each page fault.  
    \item[-] In-Cache LFU only counts the pages inside the cache at any moment. At a page fault, the page with the least frequency is evicted to make room for the incoming requested page. 
\end{itemize}
Note that an online paging algorithm is called {\em lazy} if it only evicts one page to make room for the incoming requested page at a page fault. 
Many classic paging algorithms belong to this class and have been verified to perform well, since requests typically follow {\em locality of reference} in practice, i.e., once a page is requested, it is likely that this page will be soon requested again during the following requests. 
By definition, In-Cache LFU is lazy and PLFU is not lazy. 
However, PLFU can be reduced to In-Cache LFU without paying an extra cost.\footnote{The reduction from PLFU to In-Cache LFU is trivial by not adding back the $k^{\text{th}}$ most frequent page back inside the cache right after PLFU dealing with a page fault. 
To explain why no extra cost is incurred through this reduction, given any sequence of requests, consider when requesting a $j^{\text{th}}$ most frequent page ($j > k$): by applying PLFU, a cost of 2 is incurred for processing this page $j$; however, by applying In-Cache LFU, at most a cost of 1 is incurred to process this request and also at most a cost of 1 is incurred for the next time when one of the $k$ most frequent pages is requested.}

Another salient feature of PLFU is that, up to a factor of 2, PLFU is the optimal online paging algorithm in the stochastic input model.
This is because, PLFU produces a {\em cost rate} (per request) of $2\sum_{j > k} p_j$ and no online paging algorithm can produce a cost rate less than $\sum_{j > k} p_j$ due to the cache size constraint ($p_j$ is the probability of the $j^{\text{th}}$ most frequent page in $\pi$; for a detailed proof, see Proposition \ref{proposition-costrate-plfu}). 

\paragraph{Least-Recently-Used and marking algorithms.}
Now we elaborate LRU as an online paging algorithm belonging to a class called {\em marking}, which processes the sequence of requests in {\em phases} defined as follows. 
\begin{itemize}
    \item[-] At the beginning of each phase, all the pages inside the cache (if exists) are unmarked; the first phase starts from the first request in the given sequence. 
    \item[-] When a page is requested (in a phase), it is marked immediately.
    \item[-] At a page fault, an unmarked page in the cache is evicted to make room for the requested page; in the case that all $k$ pages in the cache are marked, then the current phase terminates, all the pages become unmarked, and the next phase starts. 
\end{itemize}
By definition, LRU is a marking algorithm: at a page fault in any phase, LRU evicts the unmarked page which is least recently requested.

One famous marking algorithm is called {\em Flush-When-Full} (FWF), which evicts all the $k$ pages from the cache when a phase terminates.
Although FWF is not realistic, studying its competitiveness provides an in-depth understanding of the class of marking algorithms. 
Besides, one another {\em randomized} marking algorithm is called {\em Marker} (Fiat et al. \cite{fiat1991competitive}), which works as follows: at a page fault in any given phase, among all the unmarked pages, one of them is evicted uniformly at random. 
Thanks to the random eviction at page fault, Marker achieves an O($\log k$) competitive ratio, while LRU and FWF are $k$-competitive.

Here we remark that the classic paging algorithm {\em First-In-First-Out} (FIFO), which evicts the page which arrives earliest (among the $k$ pages inside the cache), is not a marking algorithm by the above definition.

\paragraph{Related works.}
Extensive papers have studied the online paging problem and here we only introduce the core results from the theoretical perspective.

In the offline setting (i.e., the sequence of requests is known before processing), the optimal algorithm $\optim$ (known as Belady's rule) works as follows: at a page fault, the page to be requested farthest away in the future is evicted from the cache, to make room for the currently requested page \cite{belady1966study}. However, paging is an inherent {\em online} problem: each request has to be processed immediately at its arrival, without knowing the requests arriving in the future. Note that Belady's rule cannot be used in the online setting and as a result, many heuristic online paging algorithms have been proposed, for instance, LRU, FIFO, FWF, Marker, etc.

In the online adversarial model (i.e., each request can be arbitrarily decided by an online adversary), LRU achieves a competitive ratio of $k$ and no deterministic online algorithm can achieve a competitive ratio less than $k$ \cite{sleator1985amortized, karlin1988competitive}.
Torng \cite{torng1998unified} showed that any deterministic marking algorithm is $\Theta(k)$-competitive. 
Later, Fiat et al. \cite{fiat1991competitive} proposed a randomized marking algorithm called {\em Marker} and showed that it achieves a competitive ratio of $2H_k = O(\log k)$ ($H_k$ denotes the $k^{\text{th}}$ Harmonic number).
Besides, they also show that no online algorithm can achieve a competitive ratio better than $\Theta(\log k)$. 

In fact, although LRU is $k$-competitive, it actually performs well in practice (usually with a performance ratio of no more than 2) and hence is referred to as the golden rule. 
To provide an explanation for such a phenomenon, some works investigated the characterization of the online adversary such that $\lru$ achieves a good performance ratio \cite{panagiotou2006adequate, angelopoulos2007separation}, while some other works proposed new models to evaluate an online paging algorithm's performance, for example, loose competitiveness \cite{young1994k}, paging with locality of reference \cite{borodin1995competitive, irani1996strongly, fiat1997truly, chrobak1999lru, albers2005paging, becchetti2004modeling}, etc.

Some previous works have studied online paging problems in the stochastic input model \cite{franaszek1974some, fagin1977asymptotic, koutsoupias2000beyond, young2000line}, for the purpose of explaining why LRU works well in practice.
Franaszek and Wagner \cite{franaszek1974some} studied the performance of a class of algorithms called $H_1$ (including classic algorithms like LRU, FIFO, etc) by comparing the cost rate of $\algor$ (denoted by $Q_{\algor}$ in their work) with $Q_0 = \sum_{j > k} p_i$. 
They showed that $\frac{Q_{\algor}}{Q_0} = \Omega(\log k)$ for any algorithm $\algor \in H_1$ and particularly for LRU, 
$$\frac{Q_{\lru}}{Q_0} \le 1 + \frac{k \cdot \sum_{j \le k} p_i}{1 + (k - 1) \cdot \sum_{j > k} p_i}.$$
However, no result in their work compares $Q_{\lru}$ or $Q_0$ with $Q_{\optim}$ (the cost rate of the offline optimal solution).
Koutsoupias and Papadimitriou \cite{koutsoupias2000beyond} analyzed $\lru$'s performance on the random sequence generated by a {\em diffuse} adversary - each request being any particular page with conditional probability no more than a fixed $\varepsilon$. 
They showed that among all deterministic online paging algorithms, $\lru$ achieves the optimal competitive ratio against the diffuse adversary.
Following \cite{koutsoupias2000beyond}, Young \cite{young2000line} also studied online paging against a diffuse adversary.
He established accurate formulas (as a function of $k$ and $\varepsilon$) to upper and lower bound the performance ratios of the deterministic online lazy algorithm (like LRU, FIFO, etc) and randomized online algorithm (like Marking \cite{fiat1991competitive}, Partition \cite{mcgeoch1991strongly}, Equitable \cite{achlioptas2000competitive}, etc).
We note, however, that for skewed distributions, i.e., power-law, we have typically $\varepsilon > \frac{1}{k}$. 
In such case \cite{young2000line} implies $\Theta(\log k)$ bounds. 
Laoutaris \cite{laoutaris2007closed} studied the performance of LRU when the input follows power-law distribution and derived an analytic formula on LRU's hit ratio, yet without comparing the hit ratios of LRU and PLFU or OPT. 
Although some previous works have considered the online paging problem in the stochastic input model, nonetheless, several questions still remain open:
\begin{itemize}
    \item[-]  What is the performance ratio of LFU and LRU compared with the offline optimal solution OPT? Is there any formula (as a function of the page probabilities in the given distribution $\pi$) to upper bound their performance ratios?
    \item[-] Is there any result to prove that LFU and LRU achieve a ``good'' performance ratio (say $O(1)$) when the input follows some practical distribution like power-law?
\end{itemize}
In this paper, our results provide an affirmative answer to these questions. 

Some previous work has studied multi-core paging problem \cite{hassidim2010cache, lopez2012paging}.
L{\'o}pez-Ortiz and Salinger \cite{lopez2012paging} showed that any online algorithm achieves a competitive ratio of $\Omega(n)$ ($n$ denoting the number of requests in the sequence) for the multi-core paging problem in the adversarial model. 
They further proved that the offline version problem is APX-hard. 
In recent years, another line of works given by Agrawal et al. \cite{agrawal2020green, agrawal2021tight, agrawal2022online} considers the parallel paging problem in a multi-core environment in the online adversarial model, i.e., with parallel processors available for processing page requests. 
In this paper, we considered the multi-core paging problem in the stochastic input model and showed that both LRU and PLFU achieves O(1) ratio-of-expectations when the sequence follow power-law distribution. 

\paragraph{Organization.}
In Section \ref{section-experiment}, we present our experimental result to show that the web traces in the modern cloud systems fit into a multi-core power-law distribution model. 
In Section \ref{section-theory}, \ref{section-power-law} and \ref{section-multi-core}, we provide a theoretical explanation of the good performance of LFU or LRU-based algorithms under the power-law assumption.
We derive four formulas to upper bound the ratio-of-expectations of PLFU and LRU (see Theorem \ref{theorem-main-1}, \ref{theorem-main-2}, \ref{theorem-main-3} and \ref{theorem-main-4} in Section \ref{section-theory}), which help in proving that both PLFU and LRU achieve constant ratio-of-expectation when the random sequence follows power-law distribution or multi-core power-law distribution (see Section \ref{section-power-law} and Section \ref{section-multi-core}). 
%in Section \ref{section-khp-main}, \ref{section-general-costrate}, \ref{section-clean-pages} respectively. 
%In Section \ref{section-power-law}, we show that Theorem \ref{theorem-main-3} and theorem \ref{theorem-main-4} help in proving that both PLFU and LRU achieve constant ratio-of-expectation when the random sequence follows power-law distribution. Section \ref{section-multi-core}, we further show that the constant ratio-of-expectation is still guaranteed in a practical scenario where multi-cores working concurrently and independently to process random sequences following power-law distribution with one cache pool sharing together. 
Finally, in Section \ref{section-conclusion}, we conclude with some lower bounds on the ratio-of-expectation in the stochastic input model and some future directions.

\section{Experimental results}
\label{section-experiment}
In this section, we report the experiments on traces provided by Huawei coming from their petabyte size systems. 
Dorado systems typically have 64-256 PB of SSD storage and 16-64 TB of smart cache. 
The traces record correspond to different usage scenarios in which hundreds ($\ge 500$) virtual machines are used.
The user requests (in each trace file) are served in a multi-core environment. 
That is, at any time, $\kappa$ different virtual machines are working concurrently and independently (thus generating $\kappa$ intermixed sequence of requests) but they are sharing the same cache pool of size $k$.
When a particular page $i$ is requested (in any sequence served by any core) but not inside the cache pool, one page from the cache has to be evicted to make room for page $i$. 
Typically, the number of cores $\kappa$ is smaller than the cache size $k$. 

In order to be consistent with previous results (e.g. \cite{breslau1999web}), we need to propose a model that reconstructs Zipf-like distributions when $\kappa=1$. Thus we assume that in the considered multi-core environment, the requests in each of $\kappa$ sequences follow a power-law distribution $\pi$ with parameter $\alpha$.  
This implies that the probability $\tilde{p}_i$ that a page $i$ is requested (in at least one of the $\kappa$ sequences) at any time is equal to 
\begin{equation*}\tilde{p}_i =C \cdot ( 1 - (1 - p_i)^\kappa) \text{\  with\  } p_i = \frac{1}{L(\alpha, m)} \cdot \frac{1}{i^\alpha},\end{equation*}
where $C$ is a normalization constant. 

Since the heavy-tailed distribution of cascade size had been experimentally observed, there have been many attempts to fit models to the data. 
Unfortunately, many of them have not used adequate metrics to properly measure the fit of the distribution. 
The rudimentary characterization of power-law distribution is that large events do happen quite often (contrary to exponential distribution). 
To apply those characteristics, it is tempting to propose a metric that would somehow punish errors on the tail of the distribution. 
Another naive idea is to assume that power-law distribution is linear on a log-log plot and use linear regression, e.g., as it seems to be the case in~\cite{breslau1999web}. 
Unfortunately, as claimed in~\cite{why-kstest-for-powerlaw}, both of those methods have serious problems with variation, and virtually any distribution can be accepted by those metrics. 
The proper method of analyzing power-law distributed data should involve a goodness-of-fit test~\cite{why-kstest-for-powerlaw}. 
The most commonly used is the Kolmogorov-Smirnov test (K-S test),
\begin{equation*}
\sup_{x} | X(x) - Y(x)|
\end{equation*}
which computes the maximum difference between CDFs of real and predicted data distribution ($X(x)$ is CDF of predicted data and $Y(x)$ CDF of real data).
Applying this metric allows us to avoid serious errors. 
For example, Bild et al. \cite{ks-twitter} showed using outlined methods that the lifetime of tweets does not follow power-law distribution but in fact, it is type-II discrete Weibull distribution.
The application of such conservative metrics can dispel any doubts about which model fits better to the data. 

According to our experiments, the page distribution fits the above model closely. 
In Figure \ref{fig:traces_huawei}, eight examples are given to show that the page probabilities indeed fit in the above model: the pages are sorted in the increasing order of their frequency/probabilities in the real trace; the black line representing the accumulated page probabilities (CDF) in the real trace and the blue line representing the accumulated probabilities of $\tilde{p}_i$). 
The parameters obtained by fitting the power-law model (i.e., $\kappa=1$) are presented in Table~\ref{tab:fits}, whereas parameters of multi-core power-law are given in Table~\ref{tab:fits-multi-core}. 
As intuitively shown by the figures as well as by the numerical values of the K-S test, the multi-core model fits the data much better. 
The only exception is Trace 1, where both models fit equally well. 
We suspect that in this case only a limited number of virtual machines were running.

The final aspect worth noting is the power-law exponent, which for the simple power-law model is always smaller than $1$ and ranges between $0.669-0.8$. 
This could be seen as confirming~\cite{breslau1999web}. 
However, this is only illusory and is an artifact of using a model that fits only partially to the data -- confirm Figure~\ref{fig:lomax} to see that a straight line can only fit to part of the distribution. When a better model is used the resulting alpha is bigger than $1$ and can reach almost $1.5$.  

\begin{table}[]
    \centering
    \begin{tabular}{c|c|c|c}
       No.  & Num. requests & $\alpha$ & K-S test \\ \hline
       Trace 1 & 279820 & 0.669 & 0,039\\
       Trace 2 & 1686689 & 0.734 & 0.121\\
       Trace 3 & 1650469 & 0.736 & 0.121\\
       Trace 4 & 422968 & 0.8 & 0.121 \\
       Trace 5 & 331162 & 0.791 & 0.145 \\
       Trace 6 & 587317 & 0.588& 0.088\\
       Trace 7 & 618577 & 0.587 & 0.084 \\
       Trace 8 & 655885 & 0.582 & 0.082 \\
    \end{tabular}
    \caption{Parameters obtained by fitting power-law model to collected traces.}
    \label{tab:fits}
\end{table}

\begin{table}[]
    \centering
    \begin{tabular}{c|c|c|c}
     No.  & $\alpha$  &$\kappa$ & K-S test\\ \hline
       Trace 1 & 0.666 & 1 & 0.039\\
       Trace 2 & 1.121 & 406.3 & 0.053\\
       Trace 3 & 1.246 & 1782 & 0.039\\
       Trace 4 & 1.448 & 400 & 0.038\\
       Trace 5 & 1.349 & 487.7 & 0.047 \\
       Trace 6 & 1.416 & 25044 & 0.036\\
       Trace 7 & 1.112 & 1019.7& 0.028 \\
       Trace 8 & 1.137 & 1430.4 & 0.022\\
    \end{tabular}
    \caption{Parameters obtained by fitting multi-core power-law model to collected traces.}
    \label{tab:fits-multi-core}
\end{table}

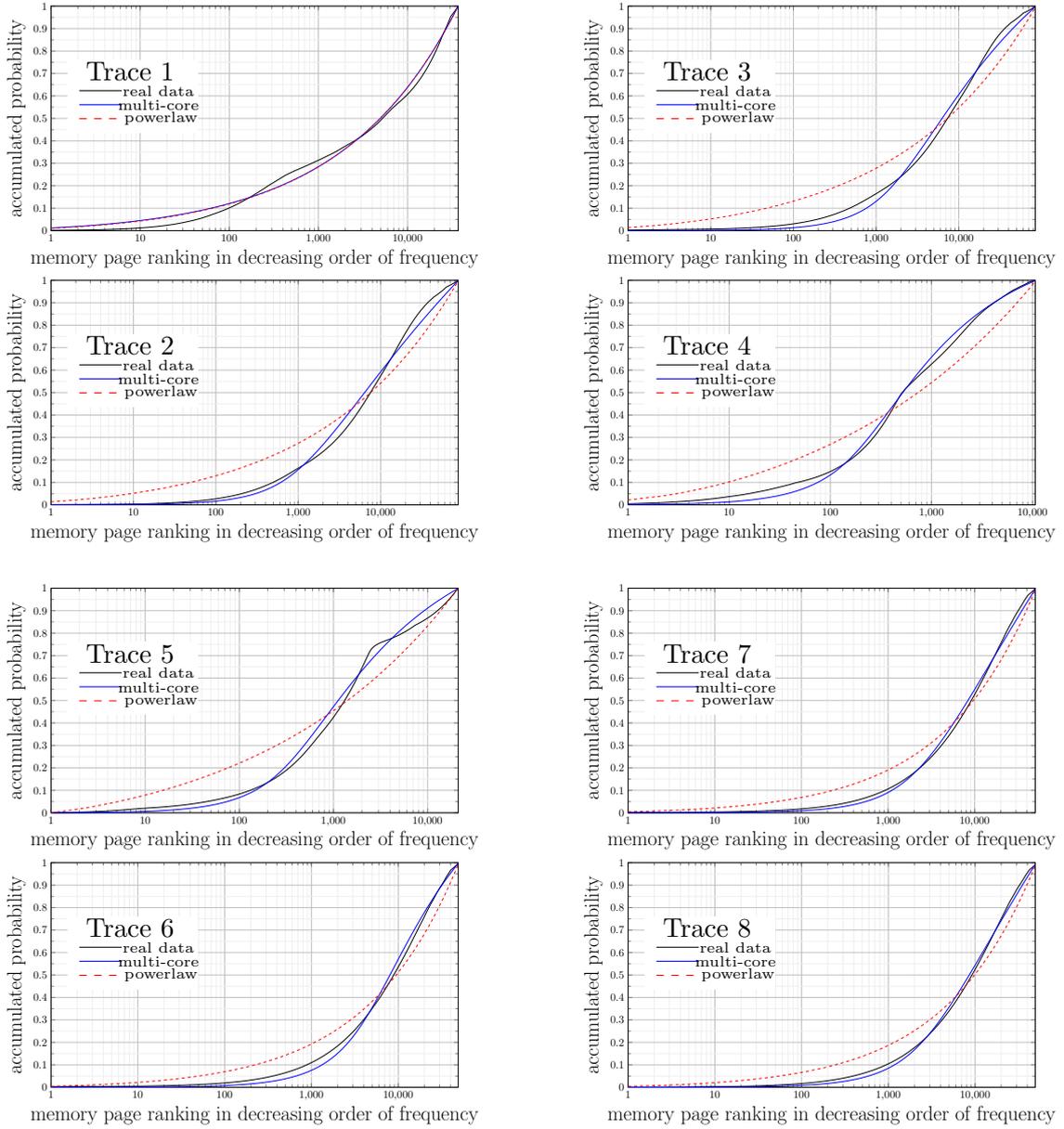
\begin{figure*}
\centering
\begin{multicols}{2}
\begin{tikzpicture}[scale=0.4] 
\begin{axis}[
    xmode=log,
    log ticks with fixed point,
    xmin = 1, xmax = 37000,
    ymin = 0, ymax = 1,
    ytick distance = 0.1,
    grid = both,
    minor tick num = 1,
    major grid style = {lightgray},
    minor grid style = {lightgray!25},
    width = \textwidth,
    height = 0.6\textwidth,
    xlabel = {memory page ranking in decreasing order of frequency},
    ylabel = {accumulated probability},
    label style={font=\huge}]

% data from huawei
\addplot[
    smooth,
    thin,
    black,
] file {trace_data/1_data.dat};

% multicore approximation
\addplot[
    smooth,
    thin,
    blue,
] file {trace_data/1_multicore.dat};

% powerlaw approximation
\addplot[
    smooth,
    thin,
    red,
    dashed,
] file {trace_data/1_powerlaw.dat};
 
\end{axis}

% legend
\fill[white] (0.6,3.5) rectangle (5.2,6.1);
\node at (2.8,5.7) {Trace 1};
\draw[black] (1,5) -- (2.5,5); 
\node at (3.8,5) {\tiny real data};
\draw[blue] (1,4.5) -- (2.5,4.5);
\node at (3.8,4.5) {\tiny multi-core};
\draw[red,dashed] (1,4) -- (2.5,4);
\node at (3.9,4) {\tiny powerlaw};

\end{tikzpicture}

% TRACE 2

\begin{tikzpicture}[scale=0.4]
\begin{axis}[
    xmode=log,
    log ticks with fixed point,
    xmin = 1, xmax = 87294,
    ymin = 0, ymax = 1,
    ytick distance = 0.1,
    grid = both,
    minor tick num = 1,
    major grid style = {lightgray},
    minor grid style = {lightgray!25},
    width = \textwidth,
    height = 0.6\textwidth,
    xlabel = {memory page ranking in decreasing order of frequency},
    ylabel = {accumulated probability},
    label style={font=\huge}]

% data from huawei
\addplot[
    smooth,
    thin,
    black,
] file {trace_data/2_data.dat};

% multicore approximation
\addplot[
    smooth,
    thin,
    blue,
] file {trace_data/2_multicore.dat};

% powerlaw approximation
\addplot[
    smooth,
    thin,
    red,
    dashed,
] file {trace_data/2_powerlaw.dat};
 
\end{axis} 

% legend
\fill[white] (0.6,3.5) rectangle (5.2,6.1);
\node at (2.8,5.7) {Trace 2};
\draw[black] (1,5) -- (2.5,5); 
\node at (3.8,5) {\tiny real data};
\draw[blue] (1,4.5) -- (2.5,4.5);
\node at (3.8,4.5) {\tiny multi-core};
\draw[red,dashed] (1,4) -- (2.5,4);
\node at (3.9,4) {\tiny powerlaw};
\end{tikzpicture}

% TRACE 3

\begin{tikzpicture}[scale=0.4] 
\begin{axis}[
    xmode=log,
    log ticks with fixed point,
    xmin = 1, xmax = 84203,
    ymin = 0, ymax = 1,
    ytick distance = 0.1,
    grid = both,
    minor tick num = 1,
    major grid style = {lightgray},
    minor grid style = {lightgray!25},
    width = \textwidth,
    height = 0.6\textwidth,
    xlabel = {memory page ranking in decreasing order of frequency},
    ylabel = {accumulated probability},
    label style={font=\huge}]

% data from huawei
\addplot[
    smooth,
    thin,
    black,
] file {trace_data/3_data.dat};

% multicore approximation
\addplot[
    smooth,
    thin,
    blue,
] file {trace_data/3_multicore.dat};

% powerlaw approximation
\addplot[
    smooth,
    thin,
    red,
    dashed,
] file {trace_data/3_powerlaw.dat};
 
\end{axis} 

% legend
\fill[white] (0.6,3.5) rectangle (5.2,6.1);
\node at (2.8,5.7) {Trace 3};
\draw[black] (1,5) -- (2.5,5); 
\node at (3.8,5) {\tiny real data};
\draw[blue] (1,4.5) -- (2.5,4.5);
\node at (3.8,4.5) {\tiny multi-core};
\draw[red,dashed] (1,4) -- (2.5,4);
\node at (3.9,4) {\tiny powerlaw};
\end{tikzpicture}

% TRACE 4

\begin{tikzpicture}[scale=0.4] 
\begin{axis}[
    xmode=log,
    log ticks with fixed point,
    xmin = 1, xmax = 10626,
    ymin = 0, ymax = 1,
    ytick distance = 0.1,
    grid = both,
    minor tick num = 1,
    major grid style = {lightgray},
    minor grid style = {lightgray!25},
    width = \textwidth,
    height = 0.6\textwidth,
    xlabel = {memory page ranking in decreasing order of frequency},
    ylabel = {accumulated probability},
    label style={font=\huge}]

% data from huawei
\addplot[
    smooth,
    thin,
    black,
] file {trace_data/4_data.dat};

% multicore approximation
\addplot[
    smooth,
    thin,
    blue,
] file {trace_data/4_multicore.dat};

% powerlaw approximation
\addplot[
    smooth,
    thin,
    red,
    dashed,
] file {trace_data/4_powerlaw.dat};
 
\end{axis} 

% legend
\fill[white] (0.6,3.5) rectangle (5.2,6.1);
\node at (2.8,5.7) {Trace 4};
\draw[black] (1,5) -- (2.5,5); 
\node at (3.8,5) {\tiny real data};
\draw[blue] (1,4.5) -- (2.5,4.5);
\node at (3.8,4.5) {\tiny multi-core};
\draw[red,dashed] (1,4) -- (2.5,4);
\node at (3.9,4) {\tiny powerlaw};
\end{tikzpicture}
\end{multicols}
\begin{multicols}{2}
% TRACE 5

\begin{tikzpicture}[scale=0.4] 
\begin{axis}[
    xmode=log,
    log ticks with fixed point,
    xmin = 1, xmax = 21090,
    ymin = 0, ymax = 1,
    ytick distance = 0.1,
    grid = both,
    minor tick num = 1,
    major grid style = {lightgray},
    minor grid style = {lightgray!25},
    width = \textwidth,
    height = 0.6\textwidth,
    xlabel = {memory page ranking in decreasing order of frequency},
    ylabel = {accumulated probability},
    label style={font=\huge}]

% data from huawei
\addplot[
    smooth,
    thin,
    black,
] file {trace_data/5_data.dat};

% multicore approximation
\addplot[
    smooth,
    thin,
    blue,
] file {trace_data/5_multicore.dat};

% powerlaw approximation
\addplot[
    smooth,
    thin,
    red,
    dashed,
] file {trace_data/5_powerlaw.dat};
 
\end{axis} 

% legend
\fill[white] (0.6,3.5) rectangle (5.2,6.1);
\node at (2.8,5.7) {Trace 5};
\draw[black] (1,5) -- (2.5,5); 
\node at (3.8,5) {\tiny real data};
\draw[blue] (1,4.5) -- (2.5,4.5);
\node at (3.8,4.5) {\tiny multi-core};
\draw[red,dashed] (1,4) -- (2.5,4);
\node at (3.9,4) {\tiny powerlaw};
\end{tikzpicture}

% TRACE 6

\begin{tikzpicture}[scale=0.4] 
\begin{axis}[
    xmode=log,
    log ticks with fixed point,
    xmin = 1, xmax = 49020,
    ymin = 0, ymax = 1,
    ytick distance = 0.1,
    grid = both,
    minor tick num = 1,
    major grid style = {lightgray},
    minor grid style = {lightgray!25},
    width = \textwidth,
    height = 0.6\textwidth,
    xlabel = {memory page ranking in decreasing order of frequency},
    ylabel = {accumulated probability},
    label style={font=\huge}]

% data from huawei
\addplot[
    smooth,
    thin,
    black,
] file {trace_data/6_data.dat};

% multicore approximation
\addplot[
    smooth,
    thin,
    blue,
] file {trace_data/6_multicore.dat};

% powerlaw approximation
\addplot[
    smooth,
    thin,
    red,
    dashed,
] file {trace_data/6_powerlaw.dat};
 
\end{axis} 

% legend
\fill[white] (0.6,3.5) rectangle (5.2,6.1);
\node at (2.8,5.7) {Trace 6};
\draw[black] (1,5) -- (2.5,5); 
\node at (3.8,5) {\tiny real data};
\draw[blue] (1,4.5) -- (2.5,4.5);
\node at (3.8,4.5) {\tiny multi-core};
\draw[red,dashed] (1,4) -- (2.5,4);
\node at (3.9,4) {\tiny powerlaw};
\end{tikzpicture}

% TRACE 7

\begin{tikzpicture}[scale=0.4] 
\begin{axis}[
    xmode=log,
    log ticks with fixed point,
    xmin = 1, xmax = 49954,
    ymin = 0, ymax = 1,
    ytick distance = 0.1,
    grid = both,
    minor tick num = 1,
    major grid style = {lightgray},
    minor grid style = {lightgray!25},
    width = \textwidth,
    height = 0.6\textwidth,
    xlabel = {memory page ranking in decreasing order of frequency},
    ylabel = {accumulated probability},
    label style={font=\huge}]

% data from huawei
\addplot[
    smooth,
    thin,
    black,
] file {trace_data/7_data.dat};

% multicore approximation
\addplot[
    smooth,
    thin,
    blue,
] file {trace_data/7_multicore.dat};

% powerlaw approximation
\addplot[
    smooth,
    thin,
    red,
    dashed,
] file {trace_data/7_powerlaw.dat};
 
\end{axis} 

% legend
\fill[white] (0.6,3.5) rectangle (5.2,6.1);
\node at (2.8,5.7) {Trace 7};
\draw[black] (1,5) -- (2.5,5); 
\node at (3.8,5) {\tiny real data};
\draw[blue] (1,4.5) -- (2.5,4.5);
\node at (3.8,4.5) {\tiny multi-core};
\draw[red,dashed] (1,4) -- (2.5,4);
\node at (3.9,4) {\tiny powerlaw};
\end{tikzpicture}

% TRACE 9 (trace 8 on the document)

\begin{tikzpicture}[scale=0.4] 
\begin{axis}[
    xmode=log,
    log ticks with fixed point,
    xmin = 1, xmax = 49376,
    ymin = 0, ymax = 1,
    ytick distance = 0.1,
    grid = both,
    minor tick num = 1,
    major grid style = {lightgray},
    minor grid style = {lightgray!25},
    width = \textwidth,
    height = 0.6\textwidth,
    xlabel = {memory page ranking in decreasing order of frequency},
    ylabel = {accumulated probability},
    label style={font=\huge}]

% data from huawei
\addplot[
    smooth,
    thin,
    black,
] file {trace_data/9_data.dat};

% multicore approximation
\addplot[
    smooth,
    thin,
    blue,
] file {trace_data/9_multicore.dat};

% powerlaw approximation
\addplot[
    smooth,
    thin,
    red,
    dashed,
] file {trace_data/9_powerlaw.dat};
 
\end{axis} 

% legend
\fill[white] (0.6,3.5) rectangle (5.2,6.1);
\node at (2.8,5.7) {Trace 8};
\draw[black] (1,5) -- (2.5,5); 
\node at (3.8,5) {\tiny real data};
\draw[blue] (1,4.5) -- (2.5,4.5);
\node at (3.8,4.5) {\tiny multi-core};
\draw[red,dashed] (1,4) -- (2.5,4);
\node at (3.9,4) {\tiny powerlaw};
\end{tikzpicture}
\end{multicols}

\caption{Fitting pages probabilities from real traces (black) using $\tilde{p}_i$ (blue) and power-law (red). }
\label{fig:traces_huawei}
\end{figure*}

\section{Theoretical results}
\label{section-theory}
In this section, we study the online paging problem in the stochastic input model (where each request is drawn independently from a page set $\M$ according to a distribution $\pi$).
We focus on deriving upper bounds on the ratio-of-expectations (RoE) of PLFU and LRU, using functions of the page probabilities of $\pi$.
We derive two formulas to upper bound $\roe(\plfu)$ and $\roe(\lru)$.
Specifically, given the distribution $\pi$, by sorting all the page probabilities in the decreasing order, i.e., $p_1 \ge p_2 \ge \dots \ge p_m$ (recalling that $m$ and $k$ denote the number of pages and the cache size, respectively) and denoting by
\begin{equation*}
    p[x:y] := \sum_{i = x}^y p_i \text{ for } 1 \le x < y \le m,
\end{equation*}
we obtain the following results. 
\begin{theorem}
\label{theorem-main-3} 
\begin{equation*}
    \roe(\lru) \le \frac{16}{\min\Big\{\frac{p[\frac{3k}{8}:\frac{k}{2}] + p[\frac{3k}{2}:m] }{p[\frac{3k}{8}:m]}, \frac{p[\frac{11k}{8}:m]}{p[\frac{k}{2}:m]}\Big\}}.
\end{equation*}
\end{theorem}
\begin{theorem}
\label{theorem-main-4}
\begin{equation*}
    \roe(\plfu) \le \frac{32}{\min\Big\{\frac{p[\frac{3k}{8}:\frac{k}{2}] + p[\frac{3k}{2}:m] }{p[\frac{3k}{8}:m]}, \frac{p[\frac{11k}{8}:m]}{p[\frac{k}{2}:m]}\Big\}}.
\end{equation*}
\end{theorem}

\begin{comment}
\begin{theorem}
\label{theorem-main-1}
$\roe(\plfu) \le 2 \cdot \sum_{i=0}^{k-1} \frac{p_{k+1}}{p_{k-i} + \dots + p_k + p_{k+1}}$.
\end{theorem}
\begin{theorem}
\label{theorem-main-2}
$\roe(\plfu) \le \frac{8}{\sum_{j > k} p_j} - 4$.
\end{theorem}
\begin{theorem}
\label{theorem-main-3}
%$\roe(\lru) \le 16 \bigg/ \min\bigg\{\frac{\sum_{i = \frac{3k}{8}}^{\frac{k}{2}} p_i }{\sum_{i = \frac{3k}{8}}^m p_i }, \frac{\sum_{i = \frac{11k}{8}}^m p_i }{\sum_{i = \frac{k}{2}}^m p_i }\bigg\}$. 
%When $m \ge 2k$, 
$\roe(\lru) \le \frac{16}{\min\bigg\{\frac{\sum_{i = \frac{3k}{8}}^{\frac{k}{2}} p_i + \sum_{i = \frac{3k}{2}}^m p_i }{\sum_{i = \frac{3k}{8}}^m p_i }, \frac{\sum_{i = \frac{11k}{8}}^m p_i }{\sum_{i = \frac{k}{2}}^m p_i }\bigg\}}$.
\end{theorem}
\begin{theorem}
\label{theorem-main-4}
%$\roe(\plfu) \le \frac{32}{\min\bigg\{\frac{\sum_{i = \frac{3k}{8}}^{\frac{k}{2}} p_i }{\sum_{i = \frac{3k}{8}}^m p_i }, \frac{\sum_{i = \frac{11k}{8}}^m p_i }{\sum_{i = \frac{k}{2}}^m p_i }\bigg\}}$. 
%When $m \ge 2k$, 
$\roe(\plfu) \le \frac{32}{\min\bigg\{\frac{\sum_{i = \frac{3k}{8}}^{\frac{k}{2}} p_i  + \sum_{i = \frac{3k}{2}}^m p_i }{\sum_{i = \frac{3k}{8}}^m p_i }, \frac{\sum_{i = \frac{11k}{8}}^m p_i }{\sum_{i = \frac{k}{2}}^m p_i }\bigg\}}$.
\end{theorem}
\end{comment}

\noindent Remark that we can assume $m \ge 2k$ with no loss: in the case $m < 2k$, this can be achieved by adding dummy pages with zero probability to appear. 

Later we show that these two results are particularly useful to prove that $\roe(\plfu)$ and $\roe(\lru)$ only depend on the power-law parameter $\alpha$ under power-law or multi-core power-law assumption. 

Furthermore, we derive two additional formulas to upper bound $\roe(\plfu)$ as follows. 
\begin{theorem}
\label{theorem-main-1}
$\roe(\plfu) \le 2 \cdot \sum_{i=2}^{k+1} \frac{p_{k+1}}{p[i:k+1]}$.
\end{theorem}
\begin{theorem}
\label{theorem-main-2}
$\roe(\plfu) \le \frac{8}{p[k+1:m]} - 4$.
\end{theorem}
\noindent Remark that depending on the given distribution $\pi$, each of the three upper bounds on $\roe(\plfu)$ can be better than the other two. \\

\noindent {\bf Technical Overview.}
We briefly introduce our techniques as follows.
To upper bound $\roe(\algor)$ for an online paging algorithm $\algor$ (not only PLFU or LRU), a typical method is to first decompose any given sequence $\sigma$ into {\em phases}, such that it is easier to both lower bound the cost of $\optim$ and also upper bound the cost of $\algor$ in each phase.
In previous works, the phases are defined recursively for the class of marking algorithms: each phase $i$ starts by seeing the first page and terminates right before seeing $k + 1$ different pages. 
By defining a page being {\em clean} in phase $i$ if it is requested within phase $i$ but not requested in previous phase $i - 1$, $\optim$ has to pay a cost at least the number of clean pages of phase $i$ during the two consecutive phases $i - 1$ and $i$.
This method of dividing $\sigma$ into phases has been used to derive the competitive ratio of the classic online paging algorithms (e.g. LRU \cite{sleator1985amortized} and Marker \cite{fiat1991competitive}) in the adversarial model.

To prove Theorem \ref{theorem-main-3}, we use the above original definition of phases used to analyse marking algorithms. 
We prove that the expected number of clean pages in each  phase (except for the last one) is at least 
\begin{equation*}
    \min\left\{\frac{p[\frac{3k}{8}:\frac{k}{2}] + p[\frac{3k}{2}:m] }{p[\frac{3k}{8}:m]}, \frac{p[\frac{11k}{8}:m]}{p[\frac{k}{2}:m]}\right\} \cdot \frac{k}{8}.
\end{equation*}
Due to the fact that any {\em marking} algorithm $\algor$ (including LRU, Marker, etc, later defined in Section \ref{section-clean-pages}) produces a cost of at most $k$ in any phase, we immediately have this upper bound for LRU. Thanks to the fact that PLFU has a cost rate $2 \cdot p[k+1:m]$ and no online algorithm can produce a cost rate less than $p[k+1:m]$, we have $\roe(\plfu) \le 2 \cdot \roe(\algor)$ for any online paging algorithm $\algor$. 
As a consequence, we have Theorem \ref{theorem-main-4}. 

To derive the two additional upper bound on $\roe(\plfu)$, we follow a similar strategy to derive $\roe(\algor)$, but using a different notion of phases, as well as being careful when taking the expectation over all possible sequences. 
More specifically, to prove Theorem \ref{theorem-main-1}, we define phases recursively as follows. 
We denote the pages with indices $1, \dots, k$ as {\em big} pages and the pages with indices $k+1, \dots, m$ as {\em small} pages. 
The current phase terminates when all the big pages have been requested (during that phase), and the next phase starts. 
In this way, $\optim$ has to pay at least the number of small pages seen in this phase. 
Using this new definition of phases, we obtain $\roe(\plfu) \le 2 \cdot \sum_{i=2}^{k+1} \frac{p_{k+1}}{p[i:k+1]}$. % (see Section \ref{section-khp-main}). 
As a corollary, this upper bound is capped by $O(\log k)$ for any distribution $\pi$, i.e., PLFU achieves a ratio-of-expectations of O($\log k$) in the stochastic input model.
See Appendix \ref{appendix-khp-main} for detailed proof of Theorem \ref{theorem-main-1}. 

Next, to prove Theorem \ref{theorem-main-2}, we use the \emph{cost rate} of a paging algorithm $\algor$, i.e., the expected cost of $\algor$ \emph{per request},
\begin{equation*}
    \costrate(\algor) = \displaystyle\overline{\lim_{n \to \infty}} \frac{\E_{\sigma}^n[\algor(\sigma)]}{n}.
\end{equation*}
Notice that $\costrate(\algor) \le c$ and $\costrate(\optim) \ge c'$ ($c, c' > 0$) implies that  $\roe(\algor) \le c / c'$. 
Since $\costrate(\plfu) = 2 \cdot p[k+1:m]$, we only need to prove that 
\begin{equation*}
    \costrate(\optim) \ge \frac{p[k+1:m]^2}{4 - 2 \cdot p[k+1:m]}.
\end{equation*}
To obtain this bound, we again use a different definition of phases (see Appendix \ref{appendix-general-costrate} for detailed proof).

\subsection{Proof of Theorem \ref{theorem-main-3} and \ref{theorem-main-4}}
\label{section-clean-pages}
In this section, we focus on proving Theorem \ref{theorem-main-3}: 
\begin{equation*}
    \roe(\lru) \le \frac{16}{\min\Big\{\frac{p[\frac{3k}{8}:\frac{k}{2}] + p[\frac{3k}{2}:m] }{p[\frac{3k}{8}:m]}, \frac{p[\frac{11k}{8}:m]}{p[\frac{k}{2}:m]}\Big\} }.
\end{equation*}
However, before we prove this theorem, we first explain why Theorem \ref{theorem-main-4} follows immediately from the fact that $\roe(\plfu) \le 2 \cdot \roe(\algor)$ for any paging algorithm $\algor$.
\begin{proposition} \label{proposition-costrate-plfu}
Given any online paging algorithm $\algor$, 
\begin{equation*}
    \roe(\plfu) \le 2 \cdot \roe(\algor).
\end{equation*}
\end{proposition}
\begin{proof}
By definition of ratio-of-expectations, we only need to prove $$\frac{\E_{\sigma}^n[\plfu(\sigma)]}{2n} \le \frac{\E_{\sigma}^n[\algor(\sigma)]}{n}$$ for any $n$.
On one hand, $$\frac{\E_{\sigma}^n[\plfu(\sigma)]}{2n} = p[k+1:m].$$
On the other hand, given any online paging algorithm $\algor$ (deterministic or randomized), before processing any request $r$, let $Q$ denote the pages inside the cache ($|Q| \le k$), then the probability to face a page fault at request $r$ is at least $1 - \sum_{i \in Q} p_i \ge p[k+1:m],$ which implies $$\frac{\E_{\sigma}^n[\algor(\sigma)]}{n} \ge \frac{n \cdot p[k+1:m]}{n} = p[k+1:m].$$
In this way, we have this proposition.
\end{proof}

\noindent Now we prove Theorem \ref{theorem-main-3}, which heavily depends on the fact that LRU belongs to a class of paging algorithms called marking. 
Recall from Section \ref{section-intro} that the marking algorithms have the following properties. 
\begin{itemize}
    \item[-] A marking algorithm processes the sequence of requests in phases. At the beginning of each phase, all the pages inside the cache (if existed) are unmarked; the first phase starts from the first request in the given sequence. 
    \item[-] When a page is requested (in a phase), it is marked immediately. At a page fault, an unmarked page in the cache is evicted to make room for the requested page; if all the $k$ pages in the cache are marked, then the current phase terminates, all the pages become unmarked and the next phase starts. 
    \item[-] Given any two consecutive phases $\phi'$, $\phi$, if a page is requested in phase $\phi$ but not requested in phase $\phi'$, then this page is called a clean page of phase $\phi$.
\end{itemize}
For any complete phase $\phi$, let $P(\phi)$ denote the $k$ pages requested in this phase and let $s(\phi)$ denote the number of clean pages in phase $\phi$, i.e., $$s(\phi) := |P(\phi) - P(\phi')|.$$
We immediately obtain the following property on $\optim$ for any given sequence of requests.
\begin{proposition}
Given any two consecutive phases $\phi'$, $\phi$, any paging algorithm (including $\optim$) produces a cost of at least $s(\phi)$ during these two phases. 
\end{proposition}
\begin{proof}
Notice that exactly $k + s(\phi)$ different pages have been requested during the phases $\phi'$ and $\phi$. 
Since the cache can contain at most $k$ different pages, for any paging algorithm, irrespective of the full cache configuration before the start of phase $\phi'$, it must have faced at least $k + s(\phi) - k = s(\phi)$ evictions during the phases $\phi'$ and $\phi$, and hence has to pay a cost at least $s(\phi)$ during the phases $\phi'$ and $\phi$.
\end{proof}
Given any marking algorithm $\algor$, let $\E_{\phi}[\algor(\phi)]$ denote the expected cost of $\algor$ produced in a complete phase $\phi$ and let $\E_{\phi}[s(\phi)]$ denote the expected number of clean pages in a complete phase $\phi$.
To upper bound $\roe(\algor)$, we only need to upper bound $\frac{\E_{\phi}[\algor(\phi)]}{\E_{\phi}[s(\phi)]}$.
Denoting by $\sigma \sim \ell$-CP if $\sigma$ is a random sequence that consists of $\ell$ complete phases, 
\begin{eqnarray} \label{deriving-roe-original}
    \roe(\algor)
    &=& \overline{\displaystyle\lim_{n \to \infty}} \frac{\E_{\sigma}^n[\algor(\sigma)]}{\E_{\sigma}^n[\optim(\sigma)]} = \overline{\displaystyle\lim_{\ell \to \infty}} \frac{\E_{\sigma}^{\ell\text{-}\cp}[\algor(\sigma)]}{\E_{\sigma}^{\ell\text{-}\cp}[\optim(\sigma)]} \nonumber \\
    &\le& 2 \cdot \overline{\displaystyle\lim_{\ell \to \infty}} \frac{\E_{\sigma}^{\ell\text{-}\cp}[\sum_\phi^{\ell} \algor(\phi)]}{\E_{\sigma}^{\ell\text{-}\cp}[\sum_\phi^{\ell} s(\phi)]} 
    = 2 \cdot \frac{\E_{\phi}[\algor(\phi)]}{\E_{\phi}[s(\phi)]}. 
\end{eqnarray}
By definition, we observe the following property. 
\begin{proposition}
\label{propo-marking-cost}
Given any marking algorithm $\algor$, we have $\E_{\phi}[\algor(\phi)] \le k$.
\end{proposition}
\begin{proof}
Once a page is marked in a phase, it stays in the cache until the end of the current phase. % - this means that no page fault will occur if an already marked page is requested.
In particular, requesting a marked page does not result in a page fault. Since the cache size is $k$, a marking algorithm has a cost of at most $k$ during any phase. 
\end{proof}
As a result, to prove Theorem \ref{theorem-main-3}, we only need to prove the following lemma.
\begin{lemma}
\label{expected-clean-pages}
\begin{equation*}
    \E_{\phi}[s(\phi)] \ge \min\Big\{\frac{p[\frac{3k}{8}:\frac{k}{2}] + p[\frac{3k}{2}:m] }{p[\frac{3k}{8}:m]}, \frac{p[\frac{11k}{8}:m]}{p[\frac{k}{2}:m]}\Big\}  \cdot \frac{k}{8}.
\end{equation*}
\end{lemma}
\begin{proof}
We (again) remark here that  w.l.o.g. we can assume that $m \ge 2k$ by adding dummy pages with zero probability to be requested. 

To prove the lemma, we show that, given any $k$ distinct pages $Q \subseteq \M$ requested in the previous phase $\phi'$, i.e., $P(\phi') = Q$, 
\begin{eqnarray}
    \E_{\phi}[s(\phi) \,|\, P(\phi') = Q] \ge \min\Big\{\frac{p[\frac{3k}{8}:\frac{k}{2}] + p[\frac{3k}{2}:m] }{p[\frac{3k}{8}:m]}, \frac{p[\frac{11k}{8}:m]}{p[\frac{k}{2}:m]}\Big\}  \cdot \frac{k}{8}. \nonumber 
\end{eqnarray}
We first partition $P(\phi')$ into two sets $P_1(\phi')$ and $P_2(\phi')$ as follows:
\begin{equation*}
    P_1(\phi') = P(\phi') \cap \Big\{1, \dots, \frac{k}{2}\Big\}, 
\end{equation*}
\begin{equation*}
    P_2(\phi') = P(\phi') \cap \Big\{\frac{k}{2} + 1, \dots, m\Big\}. 
\end{equation*}
Besides, define
\begin{equation*}
    Q_1(\phi') = \Big\{1, \dots, \frac{k}{2}\Big\} - P_1(\phi').
\end{equation*}
By definition, $|P_1(\phi')| \le \frac{k}{2}$ and $|P_2(\phi')| \ge \frac{k}{2}$.

Note that $k$ different pages are drawn from $\{1, \dots, m\}$ in phase $\phi$ according to $\pi$ without replacement. 
We can relax the constraint that the first requested page in phase $\phi$ must be a clean page (recall that a clean page is in $P(\phi)$ but not in $P(\phi')$) and this relaxation can only make $\E_{\phi}[s(\phi) \,|\, P(\phi')]$ smaller.
As a result, at least $k - |P_1(\phi')| \ge \frac{k}{2}$ pages are drawn from $Q_1(\phi') \cup \{\frac{k}{2}, \dots, m\}$ in phase $\phi$ according to $\pi$ without replacement. 
For each page $i$ drawn from $Q_1(\phi')$ in phase $\phi$, by definition, $i \notin P(\phi')$ and hence it is a clean page;
for each page $i$ drawn from $\{\frac{k}{2} + 1, \dots, m\}$ in phase $\phi$, if $i \notin P_2(\phi')$, it is also a clean page. 

Given any $k$ pages set $P(\phi)$, we have 
\begin{eqnarray}
s(\phi) = |P(\phi) - P(\phi')| \ge |P(\phi) \cap \big(Q_1(\phi') \cup \{\frac{k}{2}, \dots, m\}\big) - P(\phi')|. \nonumber
\end{eqnarray}
To lower bound $\E_{\phi}[s(\phi) \,|\, P(\phi')]$, we thus need to lower bound the expected number of clean pages amongst the pages drawn from $Q_1(\phi') \cup \{\frac{k}{2}, \dots, m\}$ in phase $\phi$.\\

\noindent We first consider the case $|Q_1(\phi')| \ge \frac{k}{4}$ (i.e., less than $\frac{k}{4}$ pages from $\{1, \dots, \frac{k}{2}\}$ are drawn in the previous phase $\phi'$).
When collecting the $j^{\text{th}}$ ($j \le \frac{k}{8}$) distinct page from $Q_1(\phi') \cup \{\frac{k}{2}, \dots, m\}$ in phase $\phi$ (with $j - 1$ different pages $R$ already drawn from $Q_1(\phi') \cup \{\frac{k}{2}, \dots, m\}$), the probability to collect a (clean) page among $Q_1(\phi')$ is 
\begin{eqnarray}
    \frac{\sum_{i \in Q_1(\phi') \cup \{\frac{k}{2}, \dots, m\} - R} p_i}{\sum_{i \in Q_1(\phi') - R} p_i + \sum_{i = \frac{k}{2} + 1}^m p_i} 
    \ge \frac{p[\frac{k}{2} - (j - 1):\frac{k}{2}] + p[\frac{k}{2} + k:m]}{p[\frac{k}{2} - (j - 1):\frac{k}{2}] + p[\frac{k}{2} + 1:m]} \ge \frac{p[\frac{3k}{8}:\frac{k}{2}] + p[\frac{3k}{2}:m]}{p[\frac{3k}{8}:m]}. \nonumber
\end{eqnarray}
(the first inequality is due to the fact that $\sum_{i \in R} p_i \le \sum_{i = 1}^{j - 1} p_i$ and $\frac{a - c}{b - c} \le \frac{a}{b}$ for $0 < c < a < b$).
Therefore, the expected clean pages number among the pages drawn from $Q_1(\phi') \cup \{\frac{k}{2} + 1, \dots, m\}$ in phase $\phi$, is at least 
\begin{equation*}
    \sum_{j = 1}^{\frac{k}{8}} \frac{p[\frac{3k}{8}:\frac{k}{2}] + p[\frac{3k}{2}:m]}{p[\frac{3k}{8}:m]} = \frac{p[\frac{3k}{8}:\frac{k}{2}] + p[\frac{3k}{2}:m]}{p[\frac{3k}{8}:m]} \cdot \frac{k}{8}.
\end{equation*}

\noindent Now we consider the case $|Q_1(\phi')| < \frac{k}{4}$ (i.e., more than $\frac{k}{4}$ pages from $\{1, \dots, \frac{k}{2}\}$ are drawn in the previous phase $\phi'$). 
In other words, less than $k - \frac{k}{4} = \frac{3k}{4}$ pages are drawn among $\{\frac{k}{2} + 1, \dots, m\}$ in the previous phase $\phi'$. 
As a result, when collecting the $j^{\text{th}}$ ($j \le \frac{k}{8}$) distinct page from $\{\frac{k}{2}, \dots, m\}$ in phase $\phi$ (with $j - 1$ different pages $R$ already drawn from $\{\frac{k}{2}, \dots, m\}$), the probability to collect a clean page is at least 
\begin{eqnarray}
    \frac{\sum_{i \in \{\frac{k}{2} + 1, \dots, m\} - P_2(\phi') - R} p_i}{\sum_{i = \frac{k}{2} + 1}^m p_i} 
    \ge \frac{p[\frac{k}{2} + \frac{3k}{4} + \frac{k}{8}:m]}{p[\frac{k}{2}:m]} = \frac{p[\frac{11k}{8}:m]}{p[\frac{k}{2}:m]}. \nonumber
\end{eqnarray}
In this way, the expected number of clean pages among the pages drawn from $\{\frac{k}{2} + 1, \dots, m\}$ in phase $\phi$, is at least $$\sum_{j = 1}^{\frac{k}{8}} \frac{p[\frac{11k}{8}:m]}{p[\frac{k}{2}:m]} = \frac{p[\frac{11k}{8}:m]}{p[\frac{k}{2}:m]} \cdot \frac{k}{8}.$$
In summary, $\E_{\phi}[s(\phi) \,|\, P(\phi')]$ is thus at least
\begin{equation*}
    \min\Big\{\frac{p[\frac{3k}{8}:\frac{k}{2}] + p[\frac{3k}{2}:m] }{p[\frac{3k}{8}:m]}, \frac{p[\frac{11k}{8}:m]}{p[\frac{k}{2}:m]}\Big\}  \cdot \frac{k}{8},
\end{equation*}
which concludes the lemma. 
\end{proof}

\section{The power-law case}
\label{section-power-law}
Recall that the web requests from an app or a website generally follow power-law distribution \cite{breslau1999web}.
Since it is popular to apply LRU or LFU (or hybrid in both of them) to maintain the cache and this type of caching replacement policies perform well in practice, our goal is to provide a theoretical explanation of this observation. 
Thanks to Theorem \ref{theorem-main-3} and \ref{theorem-main-4}, we show that $\roe(\lru) = O(1)$ and $\roe(\plfu) = O(1)$ when the random sequence follows a power-law distribution and the number of all pages is at least twice the cache size (i.e., $m \ge 2k$).\footnote{It indeed makes sense to assume $m \ge 2k$ since the web page number is way larger than the cache size in practice. }\\

\noindent {\bf Definition of power-law distribution. }
Given a power-law distribution $\pi$, by sorting all the $m$ pages in the decreasing order of their probabilities, we have
\begin{equation*}
    p_i = \frac{1}{L(\alpha, m)} \cdot \frac{1}{i^\alpha}
\end{equation*}
for each $i \in \{1, 2, \dots, m\}$.
Here,
\begin{equation*}
    L(\alpha, m) := \sum_{i = 1}^m \frac{1}{i^\alpha} = 
    \begin{cases}
        O(1) & \textit{ when } \alpha > 1;\\
        \Theta(\log m) & \textit{ when } \alpha = 1;\\
        \Theta(m^{1 - \alpha}) & \textit{ when } \alpha \in (0, 1).
    \end{cases}
\end{equation*}

\noindent {\bf LRU and PLFU achieves O(1) ratio-of-expectations in the power-law case. }
Thanks to Theorem \ref{theorem-main-3} and \ref{theorem-main-4}, we only need to show that 
\begin{equation*}
    \min\Big\{\frac{p[\frac{3k}{8}:\frac{k}{2}] + p[\frac{3k}{2}:m] }{p[\frac{3k}{8}:m]}, \frac{p[\frac{11k}{8}:m]}{p[\frac{k}{2}:m]}\Big\} := Formula
\end{equation*}
can be lower bounded by a function only depending on $\alpha$ (seen as a constant) but not $m$ or $k$. 
Indeed this is true. 
Note that given $m \ge 2k$,
\begin{equation*}
    \frac{p[\frac{3k}{8}:\frac{k}{2}] + p[\frac{3k}{2}:m] }{p[\frac{3k}{8}:m]} \ge 1 - \frac{p[\frac{k}{2}:\frac{3k}{2}]}{p[\frac{3k}{8}:2k]}.
\end{equation*}
When $\alpha > 1$, we have 
\begin{equation*}
    1 - \frac{p[\frac{k}{2}:\frac{3k}{2}]}{p[\frac{3k}{8}:2k]} \ge 1 - \frac{2^{\alpha - 1} - (\frac{2}{3})^{\alpha - 1}}{(\frac{8}{3})^{\alpha - 1} - (\frac{1}{2})^{\alpha - 1}} \text{ and }
    \frac{p[\frac{11k}{8}:m]}{p[\frac{k}{2}:m]} \ge \frac{p[\frac{12k}{8}:2k]}{p[\frac{k}{2}:2k]} \ge \frac{(\frac{4}{3})^{\alpha - 1} - 1}{4^{\alpha - 1} - 1},
\end{equation*}
which implies 
\begin{equation*}
    Formula \ge \min\bigg\{1 - \frac{2^{\alpha - 1} - (\frac{2}{3})^{\alpha - 1}}{(\frac{8}{3})^{\alpha - 1} - (\frac{1}{2})^{\alpha - 1}}, \frac{(\frac{4}{3})^{\alpha - 1} - 1}{4^{\alpha - 1} - 1}\bigg\}.
\end{equation*}
When $\alpha = 1$, we have 
\begin{equation*}
    1 - \frac{p[\frac{k}{2}:\frac{3k}{2}]}{p[\frac{3k}{8}:2k]} \ge 1 - \frac{\log 3}{\log \frac{16}{3}} \text{ and }
    \frac{p[\frac{11k}{8}:m]}{p[\frac{k}{2}:m]} 
    \ge 1 - \frac{p[\frac{k}{2}:\frac{11k}{8}]}{p[\frac{k}{2}:2k]} 
    \ge 2 - \frac{\log 11}{\log 4}, 
\end{equation*}
which implies $Formula \ge 2 - \frac{\log 11}{\log 4} > 1 - \frac{\log 3}{\log 4}$.\\

\noindent When $\alpha \in (0, 1)$, we have
\begin{equation*}
    1 - \frac{p[\frac{k}{2}:\frac{3k}{2}]}{p[\frac{3k}{8}:2k]} \ge 1 - \frac{3^{1 - \alpha} - 1}{4^{1 - \alpha} - 1} \text{ and }
    \frac{p[\frac{11k}{8}:m]}{p[\frac{k}{2}:m]} \ge 1 - \frac{p[\frac{k}{2}:\frac{11k}{8}]}{p[\frac{k}{2}:2k]} > 1 - \frac{3^{1-\alpha} - 1}{4^{1-\alpha} - 1}, 
\end{equation*}
which implies $Formula \ge 1 - \frac{3^{1-\alpha} - 1}{4^{1-\alpha} - 1}$.

\section{The multi-core power-law case}
\label{section-multi-core}
Now we show that LRU and PLFU achieve constant ratio-of-expectations (i.e., only depending on the power-law parameter $\alpha$) in the multi-core power-law case, where the probability $\tilde{p}_i$ that a page $i$ is requested (in at least one of the $\kappa$ sequences) at any time is equal to 
\begin{equation*}
    \tilde{p}_i = C \cdot ( 1 - (1 - p_i)^\kappa) \text{\  with\  } p_i = \frac{1}{L(\alpha, m)} \cdot \frac{1}{i^\alpha},
\end{equation*}
where $C$ is a normalization constant. \\

\noindent On one hand, recall that the cost produced by any core in a complete phase is no more than $k$ (see Proposition \ref{propo-marking-cost}). On the other hand, according to Lemma \ref{expected-clean-pages}, the expected number of clean pages in a complete phase is at least 
\begin{equation*}
    \min\Big\{\frac{\tilde{p}[\frac{3k}{8}:\frac{k}{2}] + \tilde{p}[\frac{3k}{2}:m] }{\tilde{p}[\frac{3k}{8}:m]}, \frac{\tilde{p}[\frac{11k}{8}:m]}{\tilde{p}[\frac{k}{2}:m]}\Big\}  \cdot \frac{k}{8}.
\end{equation*}
Again, $\min\Big\{\frac{\tilde{p}[\frac{3k}{8}:\frac{k}{2}] + \tilde{p}[\frac{3k}{2}:m] }{\tilde{p}[\frac{3k}{8}:m]}, \frac{\tilde{p}[\frac{11k}{8}:m]}{\tilde{p}[\frac{k}{2}:m]}\Big\}$ is a function only depending on $\alpha$ %and linear in $1 / \kappa$ 
(and not depending on $k$, $m$ and $\kappa$) -- this is indeed true, since $$\kappa \cdot x \ge 1 - (1 - x)^\kappa \ge \frac{\kappa \cdot x}{1 + \kappa \cdot x} \ge \frac{\kappa \cdot x}{2}$$ for any $0 \le x \le 1 / \kappa$ (see Appendix \ref{appendix-multi-cores} for the proof of this inequality), we have
\begin{eqnarray}
    \min\Big\{\frac{\tilde{p}[\frac{3k}{8}:\frac{k}{2}] + \tilde{p}[\frac{3k}{2}:m] }{\tilde{p}[\frac{3k}{8}:m]}, \frac{\tilde{p}[\frac{11k}{8}:m]}{\tilde{p}[\frac{k}{2}:m]}\Big\}
    &\ge& \min\bigg\{\frac{\frac{\kappa}{2} \cdot (p[\frac{3k}{8}:\frac{k}{2}] + p[\frac{3k}{2}:m]) }{\kappa \cdot p[\frac{3k}{8}:m] }, \frac{\frac{\kappa}{2} \cdot p[\frac{11k}{8}:m] }{\kappa \cdot p[\frac{k}{2}:m] }\bigg\} \nonumber \\
    &=& \min\Big\{\frac{p[\frac{3k}{8}:\frac{k}{2}] + p[\frac{3k}{2}:m] }{p[\frac{3k}{8}:m]}, \frac{p[\frac{11k}{8}:m]}{p[\frac{k}{2}:m]}\Big\} \cdot \frac{1}{2}. \nonumber
\end{eqnarray}
By inequality (\ref{deriving-roe-original}), we thus have $\roe(\lru) = O(1)$ and also $\roe(\plfu) = O(1)$.

\section{Concluding remarks}
\label{section-conclusion}
In this paper, we study the online paging problem in the stochastic input model. First, we observe that in modern cloud systems Pareto type II distributions are present, whereas previously only Pareto type I distributions were reported. We provide an model explaining this behaviour which we call multi-core power-law. In our opinion, this indicates that stochastic data for testing caching strategies shall include this more general model. Second, we provide theoretical explanation for the good performance of LRU and LFU-based online paging algorithms for data following power-law and multi-core power-law distribution. Up to our knowledge this is the first known throughout proof of this fact for these types of stochastic distributions.

We conclude our work by discussing some lower bounds on the ratio-of-expectations for stochastic online paging and some future directions.

\paragraph{$O(\log k)$ cannot be further improved even in the stochastic input model. }
In the arbitrary input model, the randomized online algorithm $\marker$ is $O(\log k)$-competitive and no online algorithm can achieve a competitive ratio less than $H_k = \Omega(\log k)$ \cite{fiat1991competitive}.
In other words, $O(\log k)$ is the best performance ratio for an online paging algorithm in the adversarial model.

In the stochastic input model, does there exist an online paging algorithm achieving ratio-of-expectation better than $O(\log k)$?
Unfortunately, $O(\log k)$ remains the best performance ratio for an online paging algorithm: when the number of all pages is $k + 1$ and the distribution is uniform (i.e., $m = k + 1$ and $p_1 = p_2 = \dots = p_k = p_{k+1} = \frac{1}{k + 1}$), any online paging algorithm $\algor$ achieves $\roe(\algor) = H_k = \Theta(\log k)$. 

To explain this, we again divide any given random sequence into phases recursively using the method when defining the class of marking algorithms.
In this way, each phase terminates right before seeing $k + 1$ pages, i.e., all pages.
Note that the offline optimal solution $\optim$ only produces a cost of 1 at the beginning of each phase. 
However, in the stochastic online situation, any paging algorithm produces a cost rate of $\frac{1}{k + 1}$ to deal with each request.
On the other hand, similar as for the Coupon Collector's problem, the expected length of a complete phase is exactly 
\begin{equation*}
1 + \frac{1}{1 - \frac{1}{k + 1}} + \frac{1}{1 - \frac{2}{k + 1}} + \dots + (\frac{1}{1 - \frac{k}{k + 1}} - 1) = (k + 1) \cdot H_k.
\end{equation*}
This is because the probability to collect the $i^{\text{th}}$ new page in a complete phase is equal to $1 - \frac{i - 1}{k + 1}$ (since $i - 1$ pages are already at hand), and according to geometrical distribution property, the expected length to collect the $i$th new page becomes $\frac{1}{1 - \frac{i - 1}{k + 1}}$. 
Since each phase terminates right before seeing all the $k + 1$ different pages, we thus have $k + 1$ terms in the above summation. 
As a result, the expected cost of any algorithm $\algor$ in a complete phase is $\frac{1}{k + 1} \cdot (k + 1) \cdot H_k = H_k$, which establishes $\roe(\algor) = H_k = O(\log k)$. 

\paragraph{$\lru$ can be $O(\log k)$ times worse than $\plfu$. }
To prove this fact, consider when the random sequence follows a particular distribution (the number of pages is $k + 1$ with their probabilities $p_1 = \dots = p_k = \frac{1 - \varepsilon}{k}$ and $p_{k+1} = \varepsilon$, $\varepsilon = \frac{1}{k^3}$ is a very small value).
In this case, $\roe(\lru) = \Omega(\log k)$. 
However, $\roe(\plfu) < 4$, thanks to Theorem \ref{theorem-main-1}.
This means, $\lru$ can be $O(\log k)$ times worse than $\plfu$ in the stochastic model.

To prove $\roe(\lru) = \Omega(\log k)$, we again divide any given random sequence into phases recursively (such that each phase terminates right before seeing $k + 1$ different pages and the next phase starts with such ($k + 1$)-th page being the first page to see). Since $\optim$ only produces a cost of 1 at the beginning of each phase, we thus only need to show that the expected number of page faults produced by $\lru$ during two consecutive complete phases is $\Omega(\log k)$. 

Given the (random) sequence, let $\phi_1\phi_2$ denote any two consecutive complete phases. 
Let $t_0$ denote the first time in $\phi_1\phi_2$ when the small page $k + 1$ is requested. 
By phase definition, all the pages need to be requested within $\phi_1\phi_2$ - as such, the small page $k + 1$ is either requested during $\phi_1$ or becomes the first request of $\phi_2$.
Let $t_i \ge t_0$ denote the first time after $t_0$ when $i$ big pages have been requested. 
Again, by phase definition, each $t_i$ ($i \in \{1, \dots, k - 1\}$) happens during $\phi_1\phi_2$. 
Since there are in total $k + 1$ pages but the cache size is only $k$, at any time $\in [t_0, t_{k-1}]$ there is exactly one page not in the cache. 
Such a page is referred to as the \emph{missing} page. Note that $\lru$ faces a page fault whenever the missing page is requested. 
Besides, during $[t_0, t_{k-1}]$, according to $\lru$'s definition, the small page $k + 1$ resides in the cache - this is because, such a small page is placed in the cache at time $t_0$ and for every page fault arising during $[t_0, t_{k - 1}]$, the least recently used page (to be evicted) must be a big page. 
Note that, for any $1 \le i \le k$, the probability that the big page requested at time $t_i$ being the missing page (and hence $\lru$ faces a page fault), is 
\begin{equation*}
    \frac{\frac{1 - \varepsilon}{k}}{(1 - \varepsilon) - (i - 1) \cdot \frac{1 - \varepsilon}{k}} = \frac{1}{k - (i - 1)}.
\end{equation*}
By linearity of expectation, the total number of page faults during the time interval $[t_0, t_{k-1}]$ is equal to 
\begin{equation*}
    \sum_{i=1}^{k-1} \frac{1}{k+1-i} = \sum_{i=1}^{k-1} \frac{1}{i} = H_{k - 1} = \Omega(\log k).
\end{equation*}

\paragraph{Future directions.}

One future direction is to study a more complicated but practical stochastic model following power-law distribution, for example, the Markov paging model \cite{karlin2000markov} (where the probability to request a page depends on the previously requested page).
Another future direction is to consider a prefetching problem with stochastic input (especially when the requests follow power-law distribution).
That is, right after dealing with each request and before the arrival of the next request, it is allowed to swap $p \le k$ pages inside the cache with no cost incurred --- this is indeed witnessed nowadays in the multi-cores caching system, one core processing the requests and the other core working on prefetching (according to the recommendations learned from the history data). Designing a ``good'' prefetching algorithm is challenging both from theoretical and practical point of view.

\bibliographystyle{abbrv}
\bibliography{myreference}

\appendix

\section{Proof of Theorem \ref{theorem-main-1}}
\label{appendix-khp-main}
Now we prove Theorem \ref{theorem-main-1}: 
\begin{equation*}
    \roe(\plfu) \le 2 \cdot \sum_{i=2}^{k+1} \frac{p_{k+1}}{p[i:k+1]}.
\end{equation*}
In order to do so, we come up with a new way of defining phases, particularly for PLFU, which is different from the method of defining  phases for the class of marking algorithms.

Recall that the pages in $\M$ are sorted in the decreasing order of their probabilities in the given distribution $\pi$. The pages $1, 2, \ldots, k$ are called big pages, and the pages $k + 1, \dots, m$ are called small pages. Given any length-$n$ sequence of page requests $\sigma$, it is recursively partitioned into \emph{phases} as follows. 
\begin{itemize}
    \item[-] Each phase starts one time step after the previous one ends.
    \item[-] Each phase ends once all the big pages have been requested (during that phase), or when reaching the last request of the request sequence $\sigma$. 
\end{itemize}
A phase is \emph{complete} if all the big pages have been requested during that phase.
Note that (i) only the last phase of the sequence may be incomplete; (ii) the last page requested during a complete phase must be a big page. 
We consider any complete phase $\phi \subset \sigma$ and assume that it starts at time $t_0$ and ends at time $t_1$. 
We use $\sigma_t$ to denote the page requested at time t.
Let $s(\phi)$ denote the number of different small pages requested during phase $\phi$, i.e., 
\begin{equation*}
    s(\phi):= |\{i \mid \sigma_t = i, i > k, t_0\le t \le t_1\}|
\end{equation*}
and let $f(\phi)$ denote the number of times a small page is requested during phase $\phi$, i.e., 
\begin{equation*}
    f(\phi):= |\{t \mid \sigma_t > k, t_0\le t \le t_1\}|.
\end{equation*}
We have the following observation.
\begin{lemma}
$\plfu$ produces a cost of $2f(\phi)$ during any complete phase $\phi$, while any paging algorithm produces a cost at least $s(\phi)$ during such a phase $\phi$. 
\label{obs:fs}
\end{lemma}
\begin{proof}
Notice that exactly $k + s(\phi)$ different pages have been requested during phase $\phi$: $k$ big and $s(\phi)$ small. Since the cache can contain at most $k$ different pages, for any paging algorithm $\algor$ (even offline), irrespective of the cache configuration before the start of a phase $\phi$, $\algor$ must have faced at least $$k + s(\phi) - k = s(\phi)$$ evictions during such a phase. 
\end{proof}

Let $\phi$ be a complete phase, and $s$ be an integer with $s \in \{1, 2, \dots, m - k\}$. 
Let $\E_{\phi}[f(\phi) \,|\, s(\phi)=s]$ denote the expected value of $f(\phi)$ on a complete random phase\footnote{the probability of a given complete phase corresponds to its frequency in an infinite sequence $\sigma$. } $\phi$ with $s$ small pages. % requested in this phase.
To derive $\roe(\plfu) \le 2 \cdot \sum_{i=2}^{k+1} \frac{p_{k+1}}{p[i:k+1]}$, we only need to show that 
\begin{equation*}
    \frac{\E_{\phi}[f(\phi) \,|\, s(\phi) = s]}{s} \le \sum_{i = 2}^{k + 1} \frac{p_{k+1}}{p[i:k+1]}
\end{equation*}
for all $s \in \{1, 2, \dots, m - k\}$.  
This is because, denoting by $\sigma \sim \ell$-CP if $\sigma$ is a random sequence that consists of $\ell$ complete phases, we have
\begin{eqnarray} \label{inequality-main-1}
    \roe(\plfu)
    &=& \overline{\displaystyle\lim_{n \to \infty}} \frac{\E_{\sigma}^n[\plfu(\sigma)]}{\E_{\sigma}^n[\optim(\sigma)]} = \overline{\displaystyle\lim_{\ell \to \infty}} \frac{\E_{\sigma}^{\ell\text{-}\cp}[\plfu(\sigma)]}{\E_{\sigma}^{\ell\text{-}\cp}[\optim(\sigma)]} \nonumber\\
    &\le& 2 \cdot \overline{\displaystyle\lim_{\ell \to \infty}} \frac{\E_{\sigma}^{\ell\text{-}\cp}[\sum_\phi^{\ell} f(\phi)]}{\E_{\sigma}^{\ell\text{-}\cp}[\sum_\phi^{\ell} s(\phi)]} = 2 \cdot \frac{\E_{\phi}[f(\phi)]}{\E_{\phi}[s(\phi)]} \nonumber\\
    &=& 2 \cdot \frac{\sum_{s = 1}^{m - k} \E_{\phi}[f(\phi) \,|\, s(\phi)=s] \cdot \prob(s(\phi) = s)}{\sum_{s = 1}^{m - k} s(\phi) \cdot \prob(s(\phi) = s)}. 
\end{eqnarray}
Here the first inequality follows from Lemma \ref{obs:fs} on phase decomposition for the random sequences generated according to the distribution $\pi$. 

Now, we show that for all $s \in \{1, 2, \dots, m - k\}$: 
\begin{equation*}
    \frac{\mathbb{E}[f(\phi) \,|\, s(\phi) = s]}{s} \le \sum_{i = 2}^{k+1} \frac{p_{k+1}}{p[i:k+1]}.
\end{equation*}
In order to do so, we need the following lemma (proof provided later). 
\begin{lemma}
The expected number of times that a small page $j > k$ is requested during a random phase is at most $$1 + \sum_{i=2}^{k} \frac{p_j}{p[i:k] + p_j}.$$
\label{lemma:bounded_ratio}
\end{lemma}
By Lemma \ref{lemma:bounded_ratio}, for any small page $j$ requested at least once during a random complete phase $\phi$, it is requested for at most
\begin{equation*}
    1 + \sum_{i=2}^{k} \frac{p_j}{p[i:k] + p_j} \le 1 + \sum_{i=2}^{k} \frac{p_{k+1}}{p[i:k] + p_{k+1}} \le \sum_{i=2}^{k+1} \frac{p_{k+1}}{p[i:k+1]}
\end{equation*}
times.
Therefore, during a phase $\phi$ with $s$ different small pages requested, the expected number of the requests on these $s$ small pages is upper bounded by
\begin{equation*}
    \E[f(\phi) \mid s(\phi) = s] \le s \cdot \sum_{i=2}^{k+1} \frac{p_{k+1}}{p[i:k+1]}. 
\end{equation*}
By inequality (\ref{inequality-main-1}), we conclude Theorem \ref{theorem-main-1}. 

\begin{proof}(Lemma \ref{lemma:bounded_ratio})
Consider any small page $j > k$ in a random complete phase. 
Let $t_0 = 0$ and let $t_i$ ($i \in \{1, \dots, k\}$) denote the first time when $i$ different big pages have been requested. 
Let $B_i$ denote the set of big pages that have not yet been requested at time $t_i$ (obviously $|B_i| = k - i \ge 1$ for each $i \in \{0, \dots, k - 1\}$). 
Let $O_i$ denote the number of small pages $j$ being requested during time interval $T_i = [t_i, t_{i+1}]$. 
For any $i \in \{0, \dots, k - 1\}$, the random variable $O_i$ follows a geometric distribution of success probability parameter equal to 
\begin{equation*}
    \dfrac{\sum_{i'\in B_i}p_{i'}}{p_j+\sum_{i'\in B_i}p_{i'}}
\end{equation*}
which represents the probability that one of the not-yet-requested big pages $B_i$ is requested before the small page $j$. 
As such, the expected number of page $j$'s occurrences during $T_i$ is
\begin{align*}
    \E[O_i] &= \sum_{l \ge 0} l \cdot \bigg(1 - \dfrac{\sum_{i'\in B_i}p_{i'}}{p_j+\sum_{i'\in B_i}p_{i'}}\bigg)^l \cdot \dfrac{\sum_{i'\in B_i}p_{i'}}{p_j+\sum_{i'\in B_i}p_{i'}} 
    = \dfrac{p_j + \sum_{i'\in B_i}p_{i'}}{\sum_{i'\in B_i}p_{i'}} - 1 
    = \dfrac{p_j}{\sum_{i'\in B_i}p_{i'}}.
\end{align*}
Recall that $p_1 \ge p_2 \ge \dots \ge p_k \ge p_j$ and $|B_i| = k - i$.
We thus have $$\sum_{i'\in B_i} p_{i'} \ge p_{i+2} + \dots + p_k + p_j$$ for each $i \in \{0, \dots, k - 2\}$ and $$\sum_{i'\in B_{k-1}} p_{i'} \ge p_j.$$
Finally, what we intend to upper bound, i.e., the number of occurrences of the small page $j$ during the phase, is equal to 
\begin{equation*}
    \E[O_0 + \dots + O_{k-1}] \le 1 + \sum_{i = 2}^{k} \frac{p_j}{p[i:k] + p_j}.
\end{equation*}
In this way, we have this lemma. 
\end{proof}

\noindent With the help of Theorem \ref{theorem-main-1}, we have %, thanks to Theorem \ref{theorem-main-1}, we have the following.
\begin{corollary}
    $\roe(\plfu) \le 2 H_{k+1} = O(\log k)$.
\end{corollary}

\section{Proof of Theorem \ref{theorem-main-2}}
\label{appendix-general-costrate}
Now we prove Theorem \ref{theorem-main-2}: 
\begin{equation*}
    \roe(\plfu) \le \frac{8}{p[k+1:m]} - 4.
\end{equation*}
To establish this upper bound, we use the {\em cost rate} of a paging algorithm $\algor$, defined as follows.
\begin{equation*}
    \costrate(\algor) = \displaystyle\overline{\lim_{n \to \infty}} \frac{\E_{\sigma}^n[\algor(\sigma)]}{n}.
\end{equation*}
From the definitions of cost rate and ratio-of-expectations, we have the following proposition.
\begin{proposition} \label{propo-roe-costrate}
Given $\costrate(\algor) \le c$ and $\costrate(\optim) \ge c'$, we have $\roe(\algor) \le \frac{c}{c'}$.
\end{proposition}
Note that PLFU produces a cost rate of $2 \cdot p[k+1:m]$. 
By Proposition \ref{propo-roe-costrate}, to upper bound $\roe(\plfu)$, we only need to lower bound the cost rate of offline optimal algorithm, i.e., $\costrate(\optim)$. 
\begin{lemma} \label{lemma-costrate}
    For any paging algorithm $\algor$ (even offline), we have $$\costrate(\algor) \ge \frac{p[k+1:m]^2}{4 - 2\cdot p[k+1:m]}.$$
\end{lemma}
\begin{proof}
Let $k' \le m$ denote the largest integer such that 
\begin{equation*}
    \sum_{j = 1}^{k'} p_j \le 1 - \frac{p[k+1:m]}{2}. 
\end{equation*}
Since for each $j > k$, 
\begin{equation*}
    p_j \le p_k \le \frac{p[1:k]}{k} = \frac{1-p[k+1:m]}{k};
\end{equation*}
we have 
\begin{equation*}
    k' \ge (1 + \frac{p[k+1:m]}{2 \cdot p[1:k]}) \cdot k, \textit{ i.e., } 1 - \frac{k}{k'} \ge 1 - \frac{1}{1 + \frac{p[k+1:m]}{2(1-p[k+1:m])}} = \frac{p[k+1:m]}{2 - p[k+1:m]}.
\end{equation*}
To prove this lemma, the input sequence is partitioned into phases and each phase terminates when $k'$ different pages have been requested during that phase. 
In this way, any paging algorithm (even offline) must face at least $$k'- k \ge \frac{p[k+1:m]}{2(1-p[k+1:m])} \cdot k$$ page faults in each phase, due to the limitation of the cache size $k$. 

On the other hand, to upper bound the expected length of a phase, let $D_t$ denote the (random) set of different pages requested during a sub-sequence of length $t$. 
We observe that either (i) $|D_t|\ge k'$ or (ii) the next page to be requested is not in $D_t$ with probability at least $\frac{p[k+1:m]}{2}$.
Note that when $|D_t| < k'$, the probability that the next page to be requested is not in $D_t$, i.e., $$\Pr\big(|D_{t+1}|=|D_t|+1 \big) = \sum_{j\notin D_t} p_j = 1 - \sum_{j\in D_t} p_j,$$ is at least
\begin{equation*}
    1 - \sum_{j=1}^{|D_t|} p_j \ge  1 - \sum_{j=1}^{k'} p_j \ge \frac{p[k+1:m]}{2}.
\end{equation*} 

Therefore, the expected length to collect a new page in a phase can be upper bounded by $1 / \frac{p[k+1:m]}{2} = \frac{2}{p[k+1:m]}$ and hence, the expected length of the phase (to collect $k'$ different pages) is upper bounded by $k' \cdot \frac{2}{p[k+1:m]}$. 
Since any paging algorithm $\algor$ (even offline) has to pay a cost of $k' - k$ during each phase, we can thus conclude that $\algor$ produces an expected cost rate at least 
\begin{eqnarray}
    \frac{k' - k}{k' \cdot \frac{2}{p[k+1:m]}} &=& \frac{p[k+1:m]}{2} \cdot (1 - \frac{k}{k'}) \ge \frac{p[k+1:m]^2}{4 - 2p[k+1:m]}, \nonumber 
\end{eqnarray}
which concludes the proof of this lemma.
\end{proof}

Thanks to Proposition \ref{propo-roe-costrate}, Lemma \ref{lemma-costrate} and $\costrate(\plfu) = 2\cdot p[k+1:m]$, we immediately have 
\begin{equation*}
\roe(\plfu) \le \frac{2 \cdot p[k+1:m]}{\frac{p[k+1:m]^2}{4 - 2\cdot p[k+1:m]}} = \frac{8}{p[k+1:m]} - 4,
\end{equation*}
which concludes Theorem \ref{theorem-main-2}.

In fact, since any online lazy algorithm $\algor$ satisfies $\costrate(\algor) \le 1$ (recall that a lazy algorithm only evicts a page at a page fault), we have the following corollary.
\begin{corollary}
\label{coro-any-roe}
For any lazy paging algorithm $\algor$, 
\begin{equation*}
    \roe(\algor) \le \frac{4}{p[k+1:m]^2} - \frac{2}{p[k+1:m]}.
\end{equation*}
\end{corollary}
Note that when $p[k+1:m]$ is large (i.e., there exists a constant $c > 0$ such that $p[k+1:m] > c$), we can conclude from Theorem \ref{theorem-main-2} and Corollary \ref{coro-any-roe} that $\roe(\plfu) = O(1)$ and $\roe(\algor) = O(1)$ for any lazy algorithm $\algor$.

\section{Missing proofs in Section \ref{section-multi-core}}
\label{appendix-multi-cores}
Here we prove that $$\kappa \cdot x \ge 1 - (1 - x)^\kappa \ge \frac{\kappa \cdot x}{1 + \kappa \cdot x} \ge \frac{\kappa \cdot x}{2} \text{ for any } 0 \le x \le 1 / \kappa. $$
Defining $f(x) = \kappa \cdot x$ and $g(x) = 1 - (1 - x)^{\kappa}$, we have $$f(0) = g(0) = 0 \text{ and } f'(x) = \kappa \ge \kappa \cdot (1 - x)^{\kappa - 1} = g'(x) \text{ and } x \le 1 / \kappa,$$ which implies $\kappa \cdot x \ge 1 - (1 - x)^\kappa$.
Next, to prove $1 - (1 - x)^\kappa \ge \frac{\kappa \cdot x}{1 + \kappa \cdot x}$, it is equivalent to prove that $$(1 - x)^{\kappa} \le \frac{1}{1 + \kappa \cdot x}, \text{ i.e., } (1 - x)^\kappa(1 + \kappa \cdot x) \le 1 \text{ when } x \le 1 / \kappa.$$
Defining $h(x) = (1 - x)^\kappa(1 + \kappa \cdot x)$, we have $h(0) = 1$ and $h'(x) = -\kappa \cdot (1 - x)^{\kappa - 1}(\kappa + 1) \cdot x \le 0$, which implies this inequality. 
Finally, since $\frac{1}{1 + \kappa \cdot x} \ge \frac{1}{2}$ when $x \le 1 / \kappa$, we thus have $\frac{\kappa \cdot x}{1 + \kappa \cdot x} \ge \frac{\kappa \cdot x}{2}$. 

\end{document}